\documentclass[11pt]{IEEEtran}
\usepackage{cite}
\usepackage{epsfig}
\usepackage{epstopdf}
\usepackage{graphicx}
\usepackage{url}
\usepackage{amsfonts}
\usepackage{amsmath,bm,amsthm}
\usepackage{amssymb}
\usepackage{multicol}
\usepackage{times}
\usepackage{psfrag}
\usepackage{subfigure}
\usepackage{stfloats}
\usepackage{array}
\usepackage{booktabs, threeparttable}
\usepackage{setspace}
\usepackage{color}
\usepackage[utf8]{inputenc}
\usepackage{bm} 

\newtheorem{prop}{Proposition}

\usepackage{mathtools,amssymb,lipsum}
\usepackage{cuted}
\allowdisplaybreaks

\usepackage{algorithm}
\usepackage{algorithmic}

\newcommand{\dop}{_{{\rm D},l}}

\newcommand{\p}{\Omega} 
\newcommand{\Z}{{\bf Z}}

\newcommand\smallo{
  \mathchoice{{\scriptstyle\mathcal{O}}}% \displaystyle
    {{\scriptstyle\mathcal{O}}}% \textstyle
    {{\scriptscriptstyle\mathcal{O}}}% \scriptstyle
    {\scalebox{.7}{$\scriptscriptstyle\mathcal{O}$}}%\scriptscriptstyle
  }

\begin{document}
\title{Uplink Sensing Using CSI Ratio in Perceptive Mobile Networks}

\author{
{Zhitong Ni,~\IEEEmembership{Member,~IEEE},
 J. Andrew Zhang,~\IEEEmembership{Senior~Member,~IEEE}, 
 Kai Wu,~\IEEEmembership{Member,~IEEE},\\
 and  Ren Ping Liu,~\IEEEmembership{Senior~Member,~IEEE} 
 }
  \thanks{ Z. Ni, J. Andrew Zhang, K. Wu, and R. Liu are with the Global Big Data Technologies Centre, University of Technology Sydney, NSW, 2007, Australia (emails: zhitong.ni@uts.edu.au, andrew.zhang@uts.edu.au, kai.wu@uts.edu.au, renping.liu@uts.edu.au).}  
 }

\maketitle

\begin{abstract}
Uplink sensing in perceptive mobile networks (PMNs), which uses uplink communication signals for sensing the environment around a base station, faces challenging issues of clock asynchronism and the requirement of a line-of-sight (LOS) path between transmitters and receivers. The channel state information (CSI) ratio has been applied to resolve these issues, however, current research on the CSI ratio is limited to Doppler estimation in a single dynamic path. This paper proposes an advanced parameter estimation scheme that can extract multiple dynamic parameters, including Doppler frequency, angle-of-arrival (AoA), and delay, in a communication uplink channel and completes the localization of multiple moving targets. Our scheme is based on the multi-element Taylor series of the CSI ratio that converts a nonlinear function of sensing parameters to linear forms and enables the applications of traditional sensing algorithms.
Using the truncated Taylor series, we develop novel multiple-signal-classification grid searching algorithms for estimating Doppler frequencies and AoAs and use the least-square method to obtain delays. Both experimental and simulation results are provided, demonstrating that our proposed scheme can achieve good performances for sensing both single and multiple dynamic paths, without requiring the presence of a LOS path.
\end{abstract}


\section{Introduction}\label{sec-system}
Perceptive Mobile Network (PMN) \cite{lushan,lushanSvy} is a recently proposed next-generation mobile network based on  joint radar-communication technology. The concept of PMN was proposed in \cite{lushan} and then elaborated in \cite{lushanSvy}. In contrast to current communication-only mobile networks, PMNs are expected to serve as ubiquitous sensing networks while providing uncompromised mobile communication services.
Integrated sensing and communication (ISAC) shows the prospect of realizing dual-function devices with reduced cost, packed size, smart functions, and uncompromised service quality.
A key link facilitating this is that the communication channel state information (CSI) resembles the radar channel  \cite{kumari,strum1}. 

As discussed in \cite{lushanSvy}, there are three main types of sensing methods using the received communication signals in PMNs. 
They are named   uplink sensing \cite{abu2018performance,huang2021mimo,li2021outer,zheng2017super,passive10}, downlink active sensing \cite{ali2020leveraging,garcia2016location,friedlander2007waveform},  and  downlink passive sensing \cite{Cataldo20,caoning20}.
In view of hardware cost and required facility changes, uplink sensing is the most viable way for realizing radar functions in PMNs.

In uplink sensing, multiple user equipements (UEs) send uplink signals to one base station (BS) for data transmission \cite{abu2018performance}.
When the number of UEs is large enough, the targets around the BS can be completely covered and the BS can perform simultaneous data transmission and target detection.
In \cite{huang2021mimo}, the authors designed an uplink channel estimation and sensing scheme based on deep learning.
The authors in \cite{li2021outer} analyzed the Cram\'{e}r-Rao bound for the uplink ISAC and concluded that the uplink multi-path environment is beneficial for improving the radar sensing accuracy. 
Besides estimating the element-wise channel, parameter extractions, which only extract the parameters of interest from the overall channel, can also be adopted for obtaining radar channels \cite{zhang2010direction,gudelay,mehanna2015maximum}.    
Some papers discussed how to extract the parameters of the ISAC channel environment.
In \cite{kong2018joint}, the authors used a low-rank tensor metric to extract three parameters including delay, angle, and Doppler of targets.
In \cite{sanson2020high}, the authors proposed a range-and-Doppler estimation scheme based on  multiple-signal classification (MUSIC) estimators. 
These papers assumed perfect synchronization between transceivers. The synchronization is not easy to realize between BS and multiple UEs since this process can be time-consuming \cite{hyder2016zadoff}. 
When the transmitter is asynchronous with the BS, there exist timing offset (TO) and carrier-frequency offset (CFO) in the channel, which need to be removed for sensing targets \cite{IndoTrack,widar2.0,wang2022single}. 

Recently, some WiFi-sensing papers have dealt with asynchronous transceiver setups and obtained key parameters including delay, angle-of-arrival (AoA), and Doppler frequency. 
In \cite{IndoTrack}, cross-antenna cross-correlation (CACC) was applied to obtain the AoA with commodity WiFi devices. In \cite{widar2.0}, CACC was used to resolve the ranging estimation problem for passive human tracking using a single WiFi link.   In \cite{wang2022single}, the authors also applied CACC to cancel the offsets and used the modulus of received signals to obtain the parameter of a human target. The CACC operation results in mirrored parameters in the output. The authors in \cite{widar2.0} used the average signal to suppress the mirrored side product.  In \cite{niTSP}, the authors considered asynchronous PMNs and perfectly canceled the mirrored unknown parameters using a mirrored MUSIC. 
 All of these works are based on CACC operations, and they would require a fixed line-of-sight (LOS) path and other assumptions for system setups \cite{zhang2022integration}. 
 Another way to perfectly cancel the offsets is to use division/ratio, rather than the cross-correlation, between the signals (CSI) obtained on different antennas \cite{zeng2019farsense}.  The authors in \cite{zeng2020multisense} proceeded to obtain parameters of multiple targets using the CSI ratio, which would lead to an unreachable hardware requirement.
All of these WiFi-sensing-based papers could be adopted in the uplink sensing in PMNs but problems would occur, because the LOS path can be obstructed and the channel fluctuation is more severe than that in an indoor environment,  which reduces the sensing resolution and substantially raises both false alarm and miss rate \cite{Laoudias2018}. 
In the application where multiple dynamic/moving objects are needed to be detected in the PMNs, most of the previous papers cannot be used either as they can detect only one moving target.

Motivated by the fact that multiple moving targets should be separately estimated in the ISAC systems, this paper develops an uplink sensing scheme that obtains key sensing parameters, including Doppler frequency, AoA, and propagation delay, of all moving targets for localization. Under an uplink channel of PMN, we perform the uplink sensing based on the unprecedentedly employed Taylor series of the CSI ratio. Compared with CACC, the CSI ratio has no requirement for a LOS path and can extract the specific targets in movements. This work can also be used in other applications, such as WiFi sensing and indoor tracking. The main contributions of this paper are

\begin{itemize}
\item We use the Taylor series to convert the CSI ratio from a nonlinear function into a linear function, which enables us to detect the moving targets excluding the asynchronous offsets without requiring a LOS path. Even with the asynchronous offsets, the proposed method can still benefit the extraction of parameters of moving targets. We also analyze the convergence of the Taylor series of the CSI ratio.

\item We extract key parameters exclusively that belong to dynamic paths of the ISAC channel. For Doppler frequency, we reconstruct the signal variation in the temporal domain. The zero frequency is suppressed and the non-zero Doppler frequencies can be extracted in the proposed Doppler estimator.

\item We form a manifold, such that the vectorized manifold is only influenced by the AoAs and known received signals. The vectorized manifold increases the AoA resolution but is ineffective when there is only one dynamic path and the number of antennas is small. 

\item We proceed to propose a joint AoA and delay estimator for one dynamic path. We demonstrate the dynamic AoA and delay can be obtained as long as the overall static component is given. 

\item We propose an estimator for dynamic delays. Multiple dynamic delays are estimated individually in the estimation range, which increases the accuracy of the delay estimates.   
\end{itemize}

Notations: $\rm\bf a$ denotes a vector, $\rm\bf A$ denotes a matrix, italic English letters like $N$ and lower-case Greek letters like $\alpha$ are a scalar.
$  {\rm\bf A }^T, {\rm\bf A }^H, {\rm\bf A }^*$, ${\rm\bf A }^{-1}$, and ${\rm\bf A }^\dag$ represent   transpose, conjugate transpose, conjugate, inverse and pseudo inverse, respectively.   $\|{\rm\bf A }\|_F$ is Frobenius norm of a matrix. 

\section{System and Channel Models}\label{SyS}

We consider the uplink communication and sensing in a PMN, as shown in Fig. \ref{model}.  
Multiple static UEs communicate with one static BS that uses received uplink signals for both communication and sensing. Each UE has one antenna. The BS uses a uniform linear array (ULA) of $N$ antennas.  
The uplink channel between the BS receiver and the UE's transmitter has multiple paths including both static and dynamic ones. 
The static paths refer to the LOS path, the paths reflected by static objects, and the ones that have negligible moving speed. The Doppler frequencies of static paths are assumed to be zeros. The dynamic paths are reflected by moving objects, such as vehicles.  The Doppler frequencies of dynamic paths are non-zeros and cause temporal phase variations in CSI.  
Since all UEs are assumed to be static, the uplink channel mainly consists of static paths and probably has several dynamic paths. In this paper, we treat these two types of paths differently and focus on estimating the parameters of dynamic paths solely.

\begin{figure}[t]
    \centering
    \includegraphics[scale=0.5]{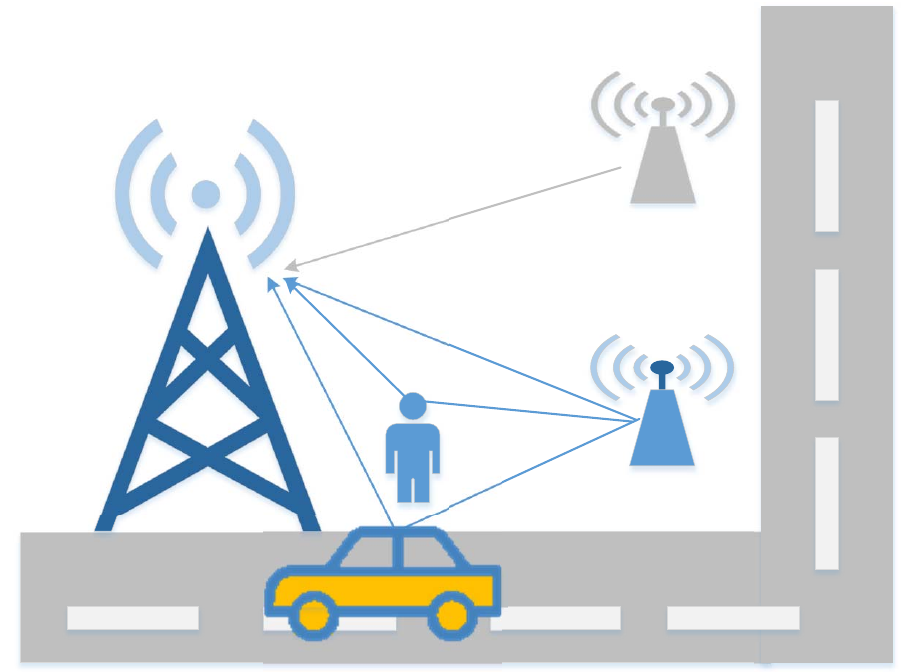}
    \caption{Illustration of uplink communication and sensing.}
    \label{model}
\end{figure}

Although logical channels are used in mobile networks and signals are transmitted in well-defined timeslots, we adopt a simplified packet structure to generate transmitted signals. In each packet, training symbols, denoted as preambles, are followed by a sequence of data symbols. Orthogonal frequency-division multiplexing (OFDM) modulation is applied across the whole packet. The data symbols can be empty if the packet is a demodulation reference signal (DMRS). For both preamble and data symbols, each of them has $G$ subcarriers with a subcarrier interval of $1/T$, where $T$ denotes the length of an OFDM symbol. Each of the OFDM symbols is prepended by a cyclic prefix (CP) of period
$T_{\rm C}$. The $m$th transmitted packet at the UE's baseband can be expressed as \cite{gudelay,strum13}
\begin{align}\label{preamble}
s(t|m)= &\sum\limits_{g=0}^{G-1} \exp\left(j2\pi g\frac tT\right){\rm rect}\left(\frac{t}{T+T_{\rm C}}\right)  x[m,g],\notag\\
&m\in\{0,\cdots,M-1\}, g\in\{0,\cdots,G-1\},
\end{align}
where $x[m, g]$ is a preamble transmitted on the $g$th subcarrier of the $m$th OFDM packet and ${\rm rect}\left(\frac{t}{T+T_{\rm C}}\right)$ denotes a rectangular window of length $T+T_{\rm C}$. For notational simplicity, we let one specific UE occupy the whole frequency band of $G/T$ and omit the index related to different UEs, but the proposed scheme can be readily applied to the case of multiple UEs by using other subcarrier assignment \cite{sit}.
This packet structure is generalized and can be used to represent signals in many wireless devices, such as WiFi and Bluetooth, in addition to mobile networks. Therefore, the scheme presented in this paper can also be applied to all these systems.
 
In this paper, we assume there are $L_S$ static paths and $L$ dynamic paths. Without loss of generality, we let the first $L$ paths, $1 \leq l\leq L$, be dynamic ones, and the rest $L_S$ paths, $L+1\leq l \leq L+L_S$, be static ones. 
Let $\alpha_l$, $f\dop$, $\tau_l$ and $\theta_l$ denote the complex channel gain, Doppler frequency, delay (propagation delay), and AoA of the $l$th path, $(1\leq l\leq L+L_S)$, respectively. 
Since there is typically no synchronization at clock level between BS and UEs, the received signal has an unknown time-varying TO, denoted as $\tau_{\rm O}[m]$, associated with the delay, even if the packet level synchronization is achieved. 
Hence, the total time delay during the signal propagation as seen by BS equals  $\tau_l + \tau_{\rm O}[m]$. 
There also exists an unknown time-varying CFO due to the asynchronous carrier frequency, denoted as $f_{\rm O}[m]$.   Assume $M$ packets are continuously transmitted with the same interval that is integer times of $T$, denoted as $T_{\rm A}$. 
The channel model is given by 
\begin{align}\label{ht}
&h(t|m)=\sum\limits_{l=1}^{L+L_S}\alpha_l \times\notag\\
&\delta\left({t-\tau_l-\tau_{\rm O}[m]-mT_{\rm A}- {(f\dop+f_{\rm O}[m]) c}t/{f_c}}\right){\bf a}(\p_l),
\end{align}
where  $\delta(t)$ is an impulse signal, $c$ is the speed of light, $f_c$ is the carrier frequency, 
${\bf a}(\p_l)=\exp[j\p_l(0,1,\cdots, N-1)]^T$, is the array response vector of size $N\times 1$, with $\p_l=\frac{2\pi d}{\lambda}\sin\theta_l$, $d$ denoting the antenna interval, $\lambda$  denoting the wavelength, and $\theta_l$ denoting the AoA from the $l$th path.

The received time-domain signal corresponding to \eqref{preamble} and \eqref{ht} can be represented as \cite{widar2.0}
\begin{align}\label{rt}
&{\bf y}(t|m)=
\sum\limits_{l=1}^{L+L_S}\alpha_l \times \notag\\
&e^{j2\pi mT_{\rm A}(f\dop+f_{\rm O}[m] ) }  s(t-\tau_l-\tau_{\rm O}[m]|m ) {\bf a}(\p_l)+{\bf w}(t|m),
\end{align}
where ${\bf w}(t|m)$ is a complex additive-white-Gaussian-noise (AWGN) vector with zero mean and variance of $\sigma^2$.

Recall that we only use the preambles, $x[m,g]$, for sensing. Hence,   $x[m,g]$ is available at the BS   and  can be easily removed by multiplying $(x[m,g])^{-1}$. 
After removing CP and $x[m,g]$, we transform the time-domain signal into the frequency domain via $G$-point fast-Fourier-transform (FFT)'s. 
Referring to \eqref{rt} and neglecting the noise, the received frequency-domain signal is
\begin{align}\label{rg}
 & y_n[m,g]\notag\\
=& \mathcal F\{h(t|m)*s(t|m)\}(x[m,g])^{-1}\notag\\
=&\sum\limits_{l=1}^{L+L_S}\alpha_l e^{jn\p_l} e^{j2\pi mT_{\rm A}(f\dop+f_{\rm O}[m])}
  e^{-j2\pi \frac{g}{T}(\tau_l+\tau_{\rm O}[m] )}  \notag\\
=& \sum\limits_{l=1}^{L }\alpha_l e^{jn\p_l} e^{j2\pi mT_{\rm A}(f\dop+f_{\rm O}[m])}
  e^{-j2\pi \frac{g}{T}(\tau_l+\tau_{\rm O}[m] )} + \notag\\
  & \sum\limits_{l=L+1}^{L+L_S}\alpha_l e^{jn\p_l} e^{j2\pi mT_{\rm A} f_{\rm O}[m] }
 e^{-j2\pi \frac{g}{T}(\tau_l+\tau_{\rm O}[m])} \notag\\
\triangleq&(D_n[m,g]+S_n[g]) e^{j2\pi mT_{\rm A} f_{\rm O}[m]} e^{-j2\pi{\frac gT}\tau_{\rm O}[m]},\notag\\
&n\in\{0,\cdots,N-1\},
\end{align}
where $n$ denotes the index of antennas at BS, $*$ denotes the convolution between two signals, $\mathcal F\{\cdot\}$ is the FFT function,  $S_n[g]=\sum\limits_{l=L+1}^{L+L_S}\alpha_l e^{jn\p_l} 
 e^{-j2\pi \frac{g}{T} \tau_l }$ is the static component without the offsets, and $D_n[m,g]=\sum\limits_{l=1}^{L}\alpha_l e^{jn\p_l} e^{j2\pi mT_{\rm A} f\dop } e^{-j2\pi \frac{g}{T} \tau_l }$ is the dynamic component without the offsets. Note that the Doppler frequencies of the static paths are zeros.    The parameters to be estimated are the dynamic ones, i.e., $\tau_l$, $\p_l$, and $f\dop$, $1\leq l\leq L$.

\section{Taylor Series of CSI Ratio}
 
In typical cases where the number of dynamic paths is much smaller than that of static paths, it would be convenient to extract the dynamic parameters without knowing  the static paths as only the moving targets are  of interest to radar.  
The CSI ratio enables such a requirement. 
Referring to \eqref{rg} and neglecting the noise term, the CSI ratio between the $n$th antenna and the $(n-q)$th antenna is given by
\begin{align}\label{CSIR}
&\xi_{n,n-q}[m,g]\notag\\
=&\frac{y_n[m,g]}{y_{n-q}[m,g]} \notag\\
=& \frac{(S_n[g] + D_n[m,g]) e^{j2\pi mT_{\rm A} f_{\rm O}[m]} e^{-j2\pi{\frac gT}\tau_{\rm O}[m]}}{ (S_{n-q}[g] + D_{n-q}[m,g])e^{j2\pi mT_{\rm A} f_{\rm O}[m]} e^{-j2\pi{\frac gT}\tau_{\rm O}[m]} }\notag\\
=&\frac{S_n[g]+ \sum\limits_{l=1}^{L}\alpha_l e^{jn\p_l} e^{j2\pi mT_{\rm A} f\dop } e^{-j2\pi \frac{g}{T} \tau_l }}{S_{n-q}[g]+ \sum\limits_{l=1}^{L}\alpha_l e^{j(n-q)\p_l} e^{j2\pi mT_{\rm A} f\dop } e^{-j2\pi \frac{g}{T} \tau_l }} \notag\\
\triangleq  & \frac{S_n[g] +  \sum\limits_{l=1}^{L} z_{n,l}[m,g]}{S_{n-q}[g] + \sum\limits_{l=1}^{L}e^{-j q\p_l} z_{n,l}[m,g]}
\triangleq   f(\Z_m), \notag\\
&q \in\{n-N+1,\cdots, n\},
\end{align}
where $ z_{n,l}[m,g] =\alpha_l e^{jn\p_l} e^{j2\pi mT_{\rm A} f\dop } e^{-j2\pi \frac{g}{T} \tau_l }$ and  
$\Z_m=[ z_{n,1}[m,g],\cdots, z_{n,L}[m,g] ]^T$. Note that different from existing works 
\cite{zeng2019farsense,zeng2020multisense}, we consider a general and more complicated case where multiple dynamic and static paths are present. 
Also note $f(\Z_m)$ is related to $g, q$, and $n$ as well, while we omit these subscripts, $g$, $q$, and $n$, for the notational simplicity of the following derivations. 
From \eqref{CSIR}, it is noted that both TOs and CFOs are fully canceled.  
It shall be highlighted that the offsets can only be removed by dividing the signals across the spatial domain. Otherwise, the ratio will involve the offsets back across other domains, i.e., $g$ or $m$.

Since there are $L$ dynamic components, the CSI ratio is a multi-element function with respect to (w.r.t.) $L$ $ z_{n,l}[m,g]$'s in the temporal domain. Note that we have stacked these $L$ variables into the vector, $\Z_m$. By using the multi-element Taylor series at the $m$th packet, the CSI ratio can be represented as
\begin{align} \label{CSITy} 
 f(\Z ) 
=  &    f(\Z_{m}) + [ \nabla f(\Z_{m})]^T(\Z-\Z_m)\notag\\
&  + \frac1{2!}[(\Z-\Z_m)]^T H(\Z_m)[\Z-\Z_m] +\smallo^3(\Z),
\end{align}
where $\nabla f(\Z_{m})=\left[ \frac{\partial f}{\partial z_1 }, \cdots, \frac{\partial f}{ \partial z_L }\right]^T$  and $H(\Z_m)$ is
\begin{align} 
   H(\Z_m) = \left[\begin{array}{ccc}
\frac{\partial^2f}{\partial z_1^2}   &\cdots  &\frac{\partial^2f}{\partial z_1\partial z_L }\\
\vdots&\ddots&\vdots\\
\frac{\partial^2f}{\partial z_L\partial z_1}&\cdots&\frac{\partial^2f}{\partial z_L^2}\\
\end{array}\right],
\end{align}
where $z_l$ is the brief notation for $ z_{n,l}[m,g]$.
 Referring to \eqref{CSIR}, the 0th-order, the 1st-order, and the 2nd-order derivatives are given by
\begin{align} \label{0thder}
  f (\Z_m) = \frac{y_n[m,g] }{y_{n-q}[m,g]  }=\xi_{n,n-q}[m,g],
\end{align}
\begin{align}\label{1thder}
 & f^{(1)}(z_l) =\frac{\partial f}{\partial z_l } \notag\\
 =  &  \frac{y_{n-q}[m,g]-e^{-jq\p_l}y_n[m,g]}{y_{n-q}[m,g]^2} e^{j2\pi mT_{\rm A} f_{\rm O}[m]} e^{-j2\pi{\frac gT}\tau_{\rm O}[m]} \notag\\
 \triangleq &  h_{n,q}^{m,g}(\p_l) e^{j2\pi mT_{\rm A} f_{\rm O}[m]} e^{-j2\pi{\frac gT}\tau_{\rm O}[m]},
\end{align}
and
\begin{align}\label{2thder}
   &f^{(2)}(z_{l_1},z_{l_2})\notag\\
 = &\frac{\partial^2f}{\partial z_{l_1}\partial z_{l_2}}\notag\\ 
 \triangleq&  H_{n,q}^{m,g}(\p_{l_1},\p_{l_2}) e^{j2\pi mT_{\rm A} 2f_{\rm O}[m]} e^{-j2\pi{\frac gT}2\tau_{\rm O}[m]},\notag\\
 &l_1\in\{1,\cdots,L\}, l_2\in\{1,\cdots,L\},
\end{align}
respectively, where
\begin{align} 
   h_{n,q}^{m,g}(\p_l) = \frac{y_{n-q}[m,g]-e^{-jq\p_l}y_n[m,g]}{y_{n-q}[m,g]^2}, 
\end{align}
and
\begin{align} 
    &H_{n,q}^{m,g}(\p_{l_1},\p_{l_2}) \notag\\
	= &  \frac{ 2e^{-jq\p_{l_2}} e^{-jq\p_{l_1}}y_n[m,g]   }{y_{n-q}[m,g]^{-3}}- \frac{ e^{-jq\p_{l_2}}+e^{-jq\p_{l_1}}    }{y_{n-q}[m,g]^{-2}}.
\end{align}
 
There is an interesting phenomenon that the offsets (TOs and CFOs) are added back in the derivatives of the Taylor series, while the CSI ratio should have removed the offsets. The received signal, $y_n[m,g]$, equals $(S_n[g]+D_n[m,g])e^{j2\pi mT_{\rm A} f_{\rm O}[m]} e^{-j2\pi{\frac gT}\tau_{\rm O}[m]}$, which means that the received signal intrinsically contains those offsets. 
Hence,  the TOs are mixed with delays as long as $y_n[m,g]$ is involved in the expression of the Taylor series.

\begin{prop}
The Taylor series of CSI ratio are convergent when $| y_{n-q}[m,g]| \geq  2L \mathop{\max}\limits_l |\alpha_l| $. 
\end{prop}
See proofs in Appendix \ref{App1}.
Note that $ y_{n-q}[m,g]$ contains the static component and $|\alpha_l|$ is the path gain of dynamic paths. The power of static paths is much stronger than that of dynamic paths. Hence, the condition is satisfied almost for sure.

Substituting \eqref{0thder}, \eqref{1thder}, and \eqref{2thder} into \eqref{CSITy} and letting $\Z$ be $\Z_{m+p}$, we can obtain
\begin{align} 
  &f(\Z_{m+p})= \xi_{n,n-q}[m+p,g]\notag\\
=& f(\Z_{m }) + \notag\\
&\sum\limits_{l=1}^L   h_{n,q}^{m,g}(\p_l)e^{j2\pi mT_{\rm A} f_{\rm O}[m]} e^{-j2\pi{\frac gT}\tau_{\rm O}[m]}\Delta z_l +\notag\\
&\frac12\sum\limits_{l_1=1}^L\sum\limits_{l_2=1}^L  H_{n,q}^{m,g}(\p_{l_1},\p_{l_2})e^{j2\pi mT_{\rm A} 2f_{\rm O}[m]} e^{-j2\pi{\frac gT}2\tau_{\rm O}[m]}\notag \times \\
& \Delta z_{l_1}\Delta z_{l_2}, p\in\{1,\cdots,P\},
\end{align}
where  $\Delta z_l=  z_{n,l}[m+p,g]- z_{n,l}[m,g]\triangleq z_l\cdot(  e^{j2\pi p T_{\rm A} f\dop }-1)$ and $z_{n,l}[m,g]$ is abbreviated as $z_l$. 
Letting  $\tilde z_l$ be  $z_l e^{j2\pi mT_{\rm A}f_{\rm O}[m] } e^{-j2\pi \frac{g}{T}  \tau_{\rm O}[m]}$, we have 
\begin{align}\label{eqsrs}
  &f(\Z_{m+p})\notag\\ 
=&   f(\Z_{m}) + \sum\limits_{l=1}^Lh_{n,q}^{m,g}(\p_l) \tilde z_l  \cdot(  e^{j2\pi p T_{\rm A} f\dop }-1)+\notag\\
 &\frac12\sum\limits_{l_1=1}^L\sum\limits_{l_2=1}^LH_{n,q}^{m,g}(\p_{l_1},\p_{l_2})  \tilde z_{l_1}^2 \tilde z_{l_2}^2 \times \notag\\
  &(  e^{j2\pi p  T_{\rm A} f_{{\rm D},l_1} }-1)\cdot(  e^{j2\pi p  T_{\rm A} f_{{\rm D},l_2} }-1).
\end{align}
Note that \eqref{eqsrs} can be transformed into $f(\Z_{m+p})-f(\Z_{m})$ that denotes the difference of CSI-ratio (D-CSIR), denoted as
\begin{align}\label{eqsrs1}
&\psi_{n,q}[m,p,g]\notag\\
=&f(\Z_{m+p})-f(\Z_m)\notag\\
  =&\xi_{n,n-q}[m+p,g]-\xi_{n,n-q}[m,g] \notag\\
 &n\in\{0,\cdots, N-1\}, q\in\{n-N+1,\cdots, n\}. 
\end{align}

Regarding the samples of D-CSIR, $f(\Z_{m+p})-f(\Z_{m})$, the Doppler frequencies of dynamic paths are clearly shown on the right-hand side of \eqref{eqsrs} and  can be retrieved by analyzing the phase variance of the D-CSIR in the temporal domain.  Unfortunately, the delays make both dynamic paths and static paths vary in the frequency domain, and hence, it is invalid to use the Taylor series w.r.t.  $g$. As for AoAs,  they do not suffer from the coupling of offsets. Traditional AoA estimation methods could be used but would involve all static paths.

\section{Dynamic Parameter Estimation }
In this section, we will propose a novel estimation scheme for obtaining Doppler frequencies, AoAs, and delays. 
Via using the CSI ratio, the scheme can exclusively extract the dynamic parameters. For the Doppler frequency, the proposed estimator is shown in section \ref{A-1}. The proposed AoA estimator is shown in section \ref{B-1} but it cannot solve the case when $L=1$ and the number of antennas is too small. Hence, we supplement a joint AoA and delay estimator in  section \ref{C-1}.  The general delay estimator is shown in section \ref{C-2}.

\subsection{Doppler Frequency Estimator}\label{A-1}

Intuitively, the Doppler frequencies can be obtained by observing the phase variance of the D-CSIR from $p=1$ to $p=P$. By assembling $p$  from 1 to $P$, we obtain a D-CSIR vector, denoted as ${\bf p}_{n,q}[m,g]= [\psi_{n,q}[m,1,g], \cdots,\psi_{n,q}[m,P,g]]^T$.
The selection of $P$ is limited. On the one hand,  the value of $P$ should be larger, such that we can use more samples of D-CSIR. On the other hand, a large value of $P$ would make the Taylor series invalid.   Since only the Doppler term is related to $p$, the MUSIC method can be  readily used to obtain the Doppler frequencies of the dynamic paths.  
Given the first two orders of the Taylor series, the non-linear CSI ratio becomes a linear function w.r.t. $(e^{j2\pi pT_{\rm A}f\dop}-1)$.

The MUSIC-type estimators require the basis vectors of ${\bf p}_{n,q}[m,g]$. Since the CSI ratio has been transformed into linear expressions via the Taylor series, the first-, the second-, and/or the higher-order Taylor series can be used to represent the basis vectors of ${\bf p}_{n,q}[m,g]$. 
The first- and second-order normalized basis vectors of ${\bf p}_{n,q}[m,g]$ are
\begin{align}\label{B1DOP}
{\bf b}_1(f) =\|{\bf b}_1(f)\|_F^{-1} \left([e^{j 2\pi  T_{\rm A} f},\cdots,e^{j2\pi PT_{\rm A}f}]^T-1 \right), 
\end{align}
and 
\begin{align} 
&{\bf b}_2(f,f') \notag\\
=& \|{\bf b}_2(f,f')\|_F^{-1} \left([e^{j 2\pi  T_{\rm A} f},\cdots,e^{j2\pi PT_{\rm A}f{\phantom'}}]^T-1\right) \times\notag\\
&\left([e^{j 2\pi   T_{\rm A} f'},\cdots,e^{j2\pi PT_{\rm A}f'}]^T-1\right), 
\end{align}
respectively, where $f$ and $f'$ are the test candidates for $f_{{\rm D},l_1}$ and $f_{{\rm D},l_2}$ in \eqref{eqsrs1}, respectively. 
With $L$ dynamic paths, there are $L$ 1st-order and $L^2$ 2nd-order basis vectors. Likewise, the number of 3rd-order basis vectors is $L^3$. 
When $L=1$, $f$ and $f'$ are the same value.
When $L\neq 1$, those harmonic vectors,  $f\neq f'$, would be inconvenient to set the test candidates. 
Note that the number of basis vectors of ${\bf p}_{n,q}[m,g]$ is no larger than $P$. Thanks to the limited number of basis vectors,  the harmonic vectors  can be dismissed.
Hence, we let all test candidates, $f$, $f'$, and .etc, be the same value in the following processing. Additionally, note that Appendix \ref{App1} has given the expression for all Taylor series. Hence, we can also construct the third-order basis vectors, denoted as ${\bf b}_3(f,f',f''),$ etc. The higher-order basis vectors can also be dismissed.
 
We fix $n=1$, $q=1$,  $g=0$, and stack the vectors,  ${\bf p}_{n,q}[m,g]$, from $m=0$  to $m=M-P-1$. 
The corresponding stacked matrix is $\bf P$. 
After dismissing the harmonic vectors and higher-order Taylor series, the required number of basis vectors equals $JL$. We select $J$, such that $JL$ is larger than the rank of $\bf P$. 
  Here, for simplicity of exposition, we assume $J=2$, which means that only the first two orders of the Taylor series are used.
Using the MUSIC method, the dynamic Doppler frequencies can be individually obtained by solving
\begin{align}\label{DopplerP}
f\dop={\rm Peak}^l\left( \frac{1}{\left\|[{\bf b}_1(f),{\bf b}_2(f,f)]^H{\bf N}_{P}\right\|_F^2}\right),
\end{align}
where ${\bf N}_P$ is the null-space of  $\bf P$ that is obtained from the left singular matrix of $\bf P$ and ${\rm Peak}^l()$ is a function that obtains $f$, such that the objective function reaches the $l$th highest peak. 
It is noted that there is a trivial solution, $f =0$, to  the problem of \eqref{DopplerP}, 
 because ${\bf b}_1(0)={\bf b}_2(0,0)={\bf 0}$. The normalization of   ${\bf b}_1(f)$ and ${\bf b}_2(f,f)$ can greatly suppress the peak generated by the trivial solution. Hence, the value of $f$ can be an arbitrary value except 0.

\subsection{AoA Estimator}\label{B-1}
For estimating AoAs of dynamic paths, similar to the proposed estimator above, we exploit the D-CSIR, such that the dynamic AoAs can be extracted solely.
From \eqref{eqsrs1}, we can observe that the AoAs are dependent on both $n$ and $q$. By assembling the samples in the spatial domain, $n$ and $q$, we also use the MUSIC-type estimators to obtain AoAs.
Note that the basis vectors in the spatial domain are related to $m$ and $g$ too, which is observed from the expressions of $h_{n,q}^{m,g}(\p_l)$ and $H_{n, q}^{m,g}(\p_{l_1},\p_{l_2})$. This indicates that $m$ and $g$ are coupled with the spatial domain and need to be fixed during the AoA estimation. Fortunately, the index, $p$, is not coupled with the spatial-domain samples. Therefore, for estimating AoAs, only measurements with different   $p$ can be stacked.

We form a spatial-domain matrix with the $(n+1,n-q+1)$th entry being the D-CSIR, $\psi_{n,q}[m,p,g] $,  $n\in\{0,\cdots,N-1\}$, $q\in\{n-N+1,\cdots, n\}$, given by
\begin{align}\label{AAA}
  &{\bf A}[m,p,g] \notag\\
= &\left[\begin{array}{ccc}
\psi_{0,0}[m,p,g]&   \cdots  & \psi_{0,-N+1}[m,p,g]\\ 
\psi_{1,1}[m,p,g]&  \cdots  & \psi_{1,-N+2}[m,p,g]\\
\vdots &\cdots&\vdots\\
\psi_{N-1,N-1}[m,p,g]&   \cdots  &\psi_{N-1,0}[m,p,g]\\
\end{array}\right].
\end{align} 
Note that the diagonal entries of ${\bf A}[m,p,g] $ are 0's since $\psi_{n,0}[m,g]=\xi_{n,n}[m+p,g]-\xi_{n,n}[m ,g]=0$. 
The number of effective entries in ${\bf A}[m,p,g]$ is $N^2-N$. With fixed $m, p$, and $g$, the matrix becomes a manifold influenced by AoAs only.

Referring to \eqref{eqsrs1}, we can obtain the basis vectors for each column of ${\bf A}[m,p,g]$. The first-order basis vector for the $(n'+1)$th column of ${\bf A}[m,p,g]$ is written as
\begin{align}\label{D1AoA}
   {\bf d}_1  (\p,n')    
 & =   \left[h_{0,-n'}^{m,g} (\p)e^{j0\p},   \cdots,  \right.\notag\\
&\left. h_{N-1,-n'+N-1}^{m,g} (\p)e^{j(N-1)\p}\right]^T, \notag\\
 & n'\in\{0,\cdots, N-1\},
\end{align} 
and the second-order basis vector  is written as
\begin{align}\label{D2AoA}
   & {\bf d}_2  (\p,\p',n')  \notag\\
 = &   \left[H_{0,-n'}^{m,g} (\p,\p')  e^{j0 },  H_{ 1,-n'+1}^{m,g} (\p,\p')  e^{j2 (\p+\p')  },   \cdots, \right.\notag\\
& \left.  H_{N-1,-n'+N-1}^{m,g} (\p,\p')  e^{j2(N-1)(\p+\p') }\right]^T, 
\end{align} 
where $\p$ and $\p'$ are test candidates for AoA estimation.

To use all columns of the manifold effectively, we vectorize the manifold into an $N^2\times 1$ vector.
Since the number of effective entries in ${\bf A}[m,p,g]$ is $N(N-1) $, the maximum number of dynamic AoAs that can be distinguished is $N(N-1)$.
The vectorized and normalized first-order basis vector is expressed as 
\begin{align}\label{DD1}
{\bf d}_1(\p)=\|{\bf d}_1(\p)\|_F^{-1}[{\bf d}_1^T(\p,0),\cdots,{\bf d}_1^T(\p,N-1)]^T,
\end{align} 
 and the vectorized and normalized second-order basis vector,
$ {\bf d}_2(\p,\p') $, is similarly obtained.

To remove the harmonic components, we let $\p' =\p$.  
The basis vectors of AoAs are dependent on $m$ and $g$ as well. Hence,  we  let $m=0$, $g=0$, and stack the vectors from $p=1$ to $p=P$ into a matrix, given by  
\begin{align}\label{barA}
 \bar {\bf A}=[{\rm vec}({\bf A}[0,1,0]),\cdots,{\rm vec}({\bf A}[0,P,0])]. 
\end{align} 
  Based on the MUSIC estimators, the AoAs are estimated by
\begin{align}\label{AoAP}
 \p_l={\rm Peak}^l \left(\frac{1}{\|[{\bf d}_1(\p),{\bf d}_2(\p,\p)]^H{\bf N}_{\bar A}\|_F^2}\right),
\end{align}
 where  ${\bf N}_{\bar A}$ is the null-space of $\bar{\bf A}$. 

There is a trivial solution to the problem of \eqref{AoAP} under a specific condition, as stated in Proposition \ref{Pro2}.
\begin{prop}\label{Pro2}
 $\p=0$ is the trivial solution to  \eqref{AoAP}  when $S_{n_1}[g] = S_{n_2}[g], \forall n_1, \forall n_2\in\{0,\cdots,N-1\}$ and $L=1$.
\end{prop}
See proofs in Appendix \ref{AppB}. 
According to Proposition \ref{Pro2}, the trivial solution happens when there is only one dynamic path and the static components, $S_n[g]$, are the same for all antennas. Such a case rarely happens in PMN because the number of antennas for a BS is large enough to assure $S_n[g]$ to be different from one another. When there are only $2, 3$ or $4$ antennas, the trivial solution could exist and influence the accuracy of the AoA estimator.

To address this issue of our proposed AoA estimator, we would like to estimate the single dynamic AoA and delay together, which will be illustrated in the next subsection.
 
\subsection{  Joint AoA and Delay Estimator for One Dynamic Path  }\label{C-1}
This subsection provides how to obtain one dynamic AoA and delay.
Referring to \eqref{eqsrs}, the delay is not explicitly expressed and only exists in $\tilde z_l$. From the expression of $\tilde z_l$ above \eqref{eqsrs}, we see that the delay is mixed with TO. 
 
\begin{prop}\label{T1}
For a single dynamic path, the dynamic delay can be obtained from the CSI ratio if and only if $S_n[g]$ is known.
\end{prop}

\begin{proof}
For a single dynamic path,  CSI ratio in \eqref{CSIR} can be rewritten as
\begin{align} 
 &\xi_{n,n-q}[m,g] \notag\\
=& \frac{S_n[g]e^{j2\pi  \frac gT x } + \alpha_1 e^{jn\p_1} e^{j2\pi mT_{\rm A} f_{{\rm D},1} }e^{j2\pi  \frac gT (x-\tau_1) }}{S_{n-q}[g]e^{j2\pi  \frac gT x } + \alpha_1 e^{j(n-q) \p_1} e^{j2\pi mT_{\rm A}  f_{{\rm D},1} } e^{j2\pi  \frac gT (x-\tau_1) } },
\end{align}
where $x$ is an arbitrary value.  Without knowing $S_n[g]$, it is impossible to differ $S_n[g]$ from $S_n[g]e^{j2\pi  \frac gT x }$ because both $S_n[g]$ and $S_n[g]e^{j2\pi  \frac gT x }$ can be generated using different values of static delays. Then, $\tau_1$ shows no difference with $\tau_1-x$. This means that no matter what $\tau_1$ becomes, the CSI ratio is the same value. 
Therefore, if the CSI ratio is used for obtaining $\tau_1$ in the dynamic path, $S_n[g]$ must be known first.
\end{proof}

According to Proposition \ref{T1}, without knowing  the static path, the dynamic delay   is impossible to be obtained from the CSI ratio. If using the ratio between adjacent $g$, the delay is entangled with TO and hence cannot be obtained either. The ratio between adjacent $m$  is too sensitive to the noise. Due to these reasons, we can see that using the CSI ratio alone is nearly impossible to obtain the delay, $\tau_1$.

Since the Doppler frequency of the dynamic path is known, all other terms excluding the Doppler-frequency term can be obtained from the CSI ratio  by solving
\begin{align}\label{CSIR2} 
\xi_{n, n-q }[m,g]=&\frac{y_n[m,g]}{y_{n-q }[m,g ]} = \frac{ S_n[g] + D_n[m,g]  }{  S_{n-q }[g ] + D_{n-q}[m,g ]  } \notag\\
=&  \frac{S_n [g]+ \alpha_1 e^{jn\p_1}e^{j2\pi mT_{\rm A} f_{{\rm D},1} } e^{-j2\pi \frac{g}{T} \tau_1 }}{S_{n-q} [g] + \alpha_1e^{j(n-q)\p_1}e^{j2\pi mT_{\rm A} f_{{\rm D},1} } e^{-j2\pi \frac{g }{T} \tau_1 }   }  \notag\\
=& \frac{S_n' [g] +  D'[m] }{S_{n-q}' [g] +   D'[m]    },   
\end{align}
where $D'[m] = e^{j2\pi mT_{\rm A} f_{{\rm D},1} }$ and $S_n'[g] = S_n[g]/\alpha_1e^{-jn\p_1}e^{ j2\pi \frac{g}{T}\tau_1}$. Note that the index $m$ and $g$ are separated in the last term of \eqref{CSIR2}. Using the previously estimated  $f_{{\rm D},1}$, $D '[m]$ becomes a known value. Then,  selecting  the CSI ratio over two or more packets, we can obtain $S_n'[g]$ by using the least-square (LS) method, i.e., 
\begin{align} 
\left[\begin{array}{c}
 S_n'[g ]\\
S_{n-q}'[g] \\
\end{array}\right]=\left[\begin{array}{ccc}
1,& -\xi_{n, q }[m_0,g] \\
1,& -\xi_{n, n-q }[m_1,g]\\ 
\end{array}\right]^{-1}\times\notag\\
\left[\begin{array}{c}
 \xi_{n,n-q }[m_0,g]  D '[m_0]-  D '[m_0] \\
 \xi_{n,n-q }[m_1,g]   D '[m_1]-  D '[m_1]\\
\end{array}\right],
\end{align}
where $m_0$ and $m_1$ are two selected packets.
It is noted that $m_0$ and $m_1$ should not be close to each other, otherwise they would result in a high noise floor. Generally, the interval between $m_0$ and $m_1$ is larger than $P$.

The phase angle of $S_n'[g]$ is equal to $(\angle S_n[g]) + 2\pi\frac gT\tau_1 -n\p_1$. According to Proposition \ref{T1}, the delay is impossible to be obtained from the CSI ratio when $S_n[g]$ is unknown.
 Note that $S_n[g]$ denotes the static component of a received signal without offsets. The phase angle of $S_n'[g]$ is time-varying due to the time-varying $\tau_1$ and $\p_1$ over a long training period. 
 
 When $ 2\pi\frac gT\tau_1 -n\p_1$ goes through a wide range within $[0,2\pi)$,  $S_n[g]$ can be seen as the expectation of $S_n'[g]$ , i.e.,
\begin{align} 
 S_n[g] \approx {\mathbb E}( S_n'[g] ).
\end{align}
Then, the dynamic delay and AoA can be extracted from the phase value of
\begin{align} 
S_n'[g]/S_n[g] =e^{ j(2\pi\frac gT\tau_1 -n\p_1)}.
\end{align}

\subsection{ Delay Estimator  }\label{C-2}
 
When AoAs can be obtained from section \ref{B-1}, we estimate delays similar to section \ref{C-1}. The main difference is that, here, the Doppler frequencies and AoAs have been obtained from section \ref{A-1} and \ref{B-1}, respectively. Hence, before estimating delays, we can match the obtained Doppler frequencies and AoAs of multiple dynamic paths.

Note   that $\bar{\bf A}$ in \eqref{barA} is related to both Doppler frequencies and AoAs. Hence, the pair of $f\dop$ and $\p_l$ can be matched by solving
\begin{align}
 ( f\dop,\p_{l'}) = \mathop{\max}\limits_{l,l'} \left|{\bf d}_1^H(\hat \p_{l'})\bar {\bf A} {\bf b}_1 (\hat f\dop)\right|, l,l'\in\{1,\cdots, L\},
\end{align}  
where ${\bf d}_1 (\hat \p_{l'})$ and ${\bf b}_1^T(\hat f\dop)$ are the basis vectors for the columns and rows of $\bar{\bf A}$, respectively, and their expressions are given by \eqref{DD1} and \eqref{B1DOP}, respectively.
 
We can obtain multiple delays from
\begin{align}
 &\xi_{n, n-q }[m,g] \notag\\
= & \frac{y_n[m,g]}{y_{n-q }[m,g ]} 
= \frac{ S_n[g] + D_n[m,g]  }{  S_{n-q }[g ] + D_{n-q}[m,g ]  } 
 \notag\\
=& \frac{S_n' [g]+ D_1'[m]+\sum\limits_{l=2}^L \frac {\alpha_l} {\alpha_1}  D_{ l}'[m] e^{jn(\p_l-\p_1)} e^{-j2\pi \frac{g}{T} (\tau_l-\tau_1) }}{S_{n-q}' [g] +D_1'[m]+ \sum\limits_{l=2}^L \frac {\alpha_l} {\alpha_1} D_{ l }'[m]  e^{jn(\p_l-\p_1)} e^{-j2\pi \frac{g}{T} (\tau_l-\tau_1) } },
\end{align}
where $D_{ l}'[m]=  e^{j2\pi mT_{\rm A} f\dop }$ and $S_n'[g] = S_n[g]/\alpha_1e^{-jn\p_1}e^{ j2\pi \frac{g}{T}\tau_1}$.
Note that $D_1'[m]$ is available.  Also using the LS method, we can obtain multiple dynamic delays from
\begin{align}\label{LSout}
&\left[  S_n'[g ],
S_{n-q}'[g],
\frac{\alpha_2}{ \alpha_1} e^{-j 2\pi \frac gT(\tau_2-\tau_1)}e^{jn(\p_l-\p_1)},
\cdots,\right.\notag\\
&\left.
\frac{\alpha_L}{ \alpha_1}e^{-j 2\pi \frac gT(\tau_L-\tau_1)}e^{jn(\p_l-\p_1)}  \right]^T \notag\\
&=\left[\begin{array}{ccc}
{\bf 1}^T \\
 -\xi_{n,n- q }[{\bf m}^T,g] \\
   D_{ 2 }'[{\bf m}^T ] -  \xi_{n, n-q }[{\bf m}^T,g]  \odot D_{ 2 }'[{\bf m} ^T] \\
  \vdots\\
    D_{ L  }'[{\bf m}^T ]-  \xi_{n, n-q }[{\bf m}^T,g]\odot  D_{ L }'[{\bf m}^T ] \\
\end{array}\right]^{T\dag}\notag\\
&\times
\left[\begin{array}{c}
 \xi_{n, n-q }[{\bf m},g] \odot D_{1 }'[{\bf m} ]- D_{1}' [{\bf m} ]\\ 
\end{array}\right],
\end{align}
where ${\bf m}=[m_0,\cdots,m_{L }]^T$ are the indexes of $(L+1)$ packets ranging from 0 to $M-1$.
From the equation above, as long as $\tau_1$ can be obtained, other delays are easy to be obtained from the $3$rd row to the last row of the LS output. We note that $\tau_1$ is obtained from $S_n'[g]$, which is realized in the same way as section \ref{C-1}.

\section{Experimental and Numerical Results}

In this section, we provide both experimental and numerical results to validate the proposed parameter estimators. We use both practical data collected by a 3-antenna  COTS WiFi device and the simulated data generated by MATLAB. 
\begin{figure*}[t] 
\centering
	\subfigure[Doppler tracing]{ 
			\includegraphics[scale=0.45]{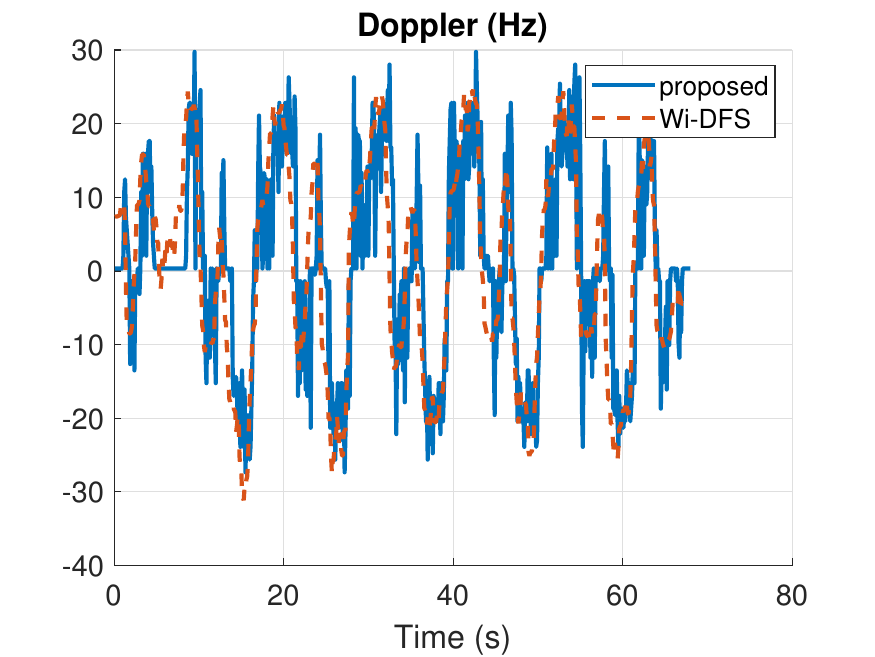} 
	}
	\subfigure[AoA tracing]{  
			\includegraphics[scale=0.45]{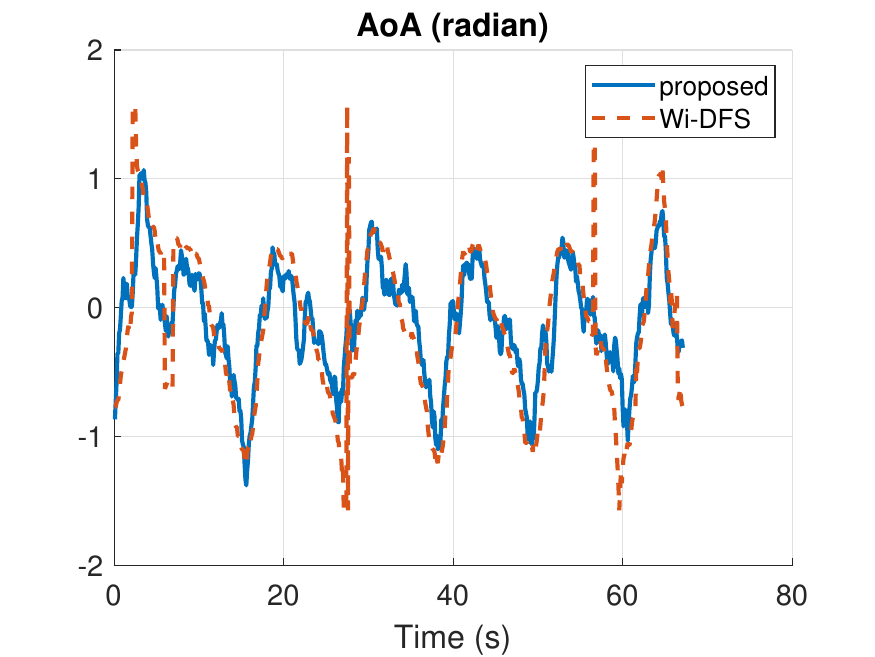} 
	}
	\subfigure[Delay tracing]{  
			\includegraphics[scale=0.45]{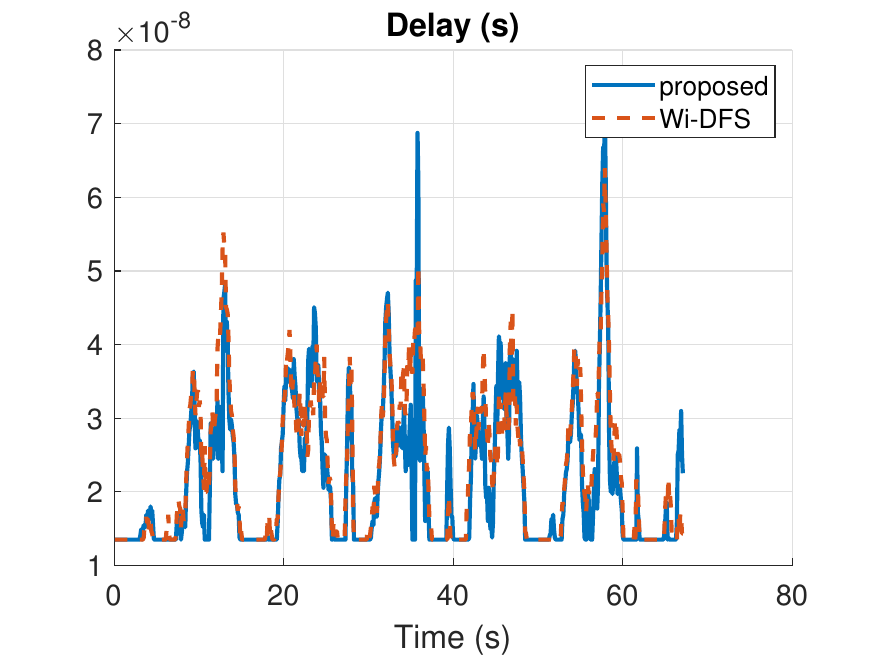} 
	}
	\subfigure[ Trajectory]{  		  
	\includegraphics[scale=0.45]{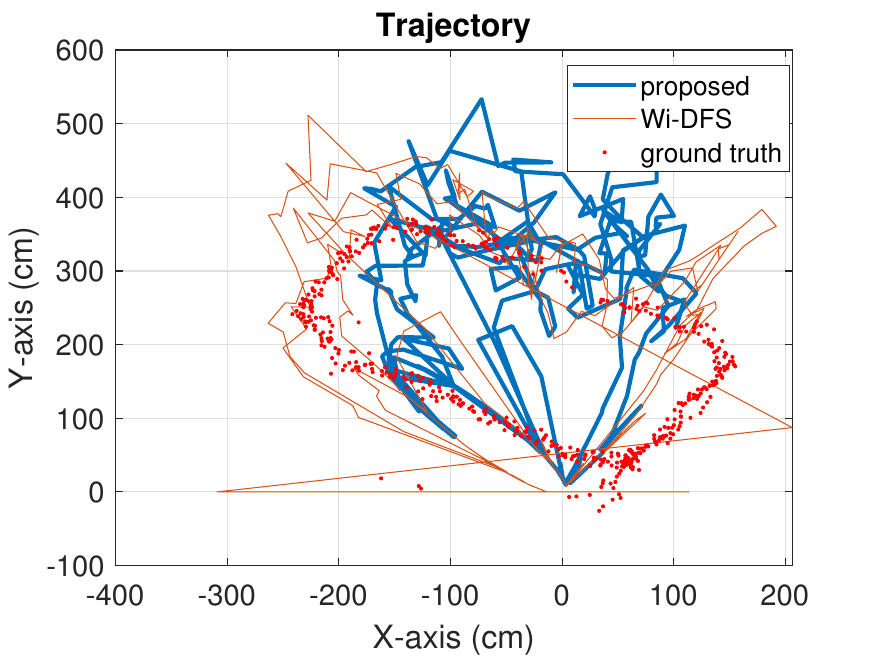} 
	}
	\caption{Tracing all three parameters and trajectory. The benchmarks include WiDFS \cite{wang2022single} and points cloud obtained by mmWave radar. }
	\label{Tracing}
\end{figure*}

\subsection{Experimental  Results}

In experimental results, the detailed setup is the same as that in \cite{wang2022single} as we implement our estimator by using their obtained raw data, where the receiver is a 3-antenna CTOS WiFi and the transmitter is one antenna that is 235 cm away from the receive antenna   $n=0$. The interval between the receive antenna  $n=0$  and  $n=1$  is 2.682 cm. The interval between the receive  antenna $n=1$ and $n=2$ is 2.251 cm. The carrier frequency is 5.32 GHz with the subcarrier interval being 312.5 kHz. The number of subcarriers, $G$, is 30. The sampling frequency, $1/T$, is 1 kHz and $T_{\rm A}$ is 0.1 s. It should be highlighted that the raw data in \cite{wang2022single} has a fixed phase diversity over different antennas due to hardware imperfection. Fortunately, such a phase diversity can be easily removed, such that the data is compatible with our estimators. The transformed received signal should be $y_n[m,g]\cdot e^{j\phi_n}$ with $\phi_n$ being the phase diversity. It should be also noted that the received signal needs to go through a low-pass filter before estimating the parameters because the raw data is mixed with high-frequency noises due to the network interface controller. We select the cutoff frequency as $60$ Hz. We let $P$ equal 30. The gap between $m_0$ and $m_1$ in our delay estimator is 30.

In Fig. \ref{Tracing}, we illustrate the Doppler frequency, AoA, delay, and trajectory obtained by processing the WiFi received signal. In practice, there is one moving human target in the indoor environment. In Fig. \ref{Tracing} (a), (b), and (c), our method shows a similar trend to the WiDFS method but provides more details. In Fig. \ref{Tracing} (b), the AoAs of the WiDFS method show fewer details and have mutated points around 28s and 58 s.
We need to point out that the delay is too small compared with $T$, which results in low accuracy of delay.
Hence, both our estimator and the WiDFS method use the Kalman filter to smooth the delay and plot the trajectory. The detailed setup of the Kalman filter can be referred to \cite{wang2022single}.
 As for the delay itself, we plot the original output of delays without the Kalman filter, as shown in Fig. \ref{Tracing} (c). We also observe that, in the practically obtained data, $S_n[g]$ can be approximately written as $|S_n[g]|e^{-\tau_0}$, where $\tau_0$ is the delay of the LOS path, as the LOS path is dominant in the received signal.
Fig. \ref{Tracing} (d) plots the trajectory calculated by the  AoA and the smoothed delay from 20s to 50s. According to the cosine law, the coordinates of the trajectory are expressed as
\begin{align}
x_{\rm traj}=d_r\sin(\theta) - d_0/2, y_{\rm traj}=d_r\cos(\theta),
\end{align}
where $d_r$, calculated by the cosine law with substituting delay and AoA, is the distance between the receive antenna $n=0$ and the human target, and $d_0=270$ cm.   The ground truth is obtained by the millimetre-wave (mmWave) radar device that is located beside the commodity WiFi.
Our obtained trajectory matches with the ground truth tightly, and we can see that the WiDFS's obtained trajectory has mutated curves due to the deviated AoAs. Most importantly, we should point out that our method does not require the existence of the LOS path even though the LOS path is dominant in the raw data, but the WiDFS method would require the existence of the LOS path to obtain both AoAs and delays.

\begin{figure}[t]
     \centering
	\subfigure[Error proportion]{  		  
	\includegraphics[scale=0.5]{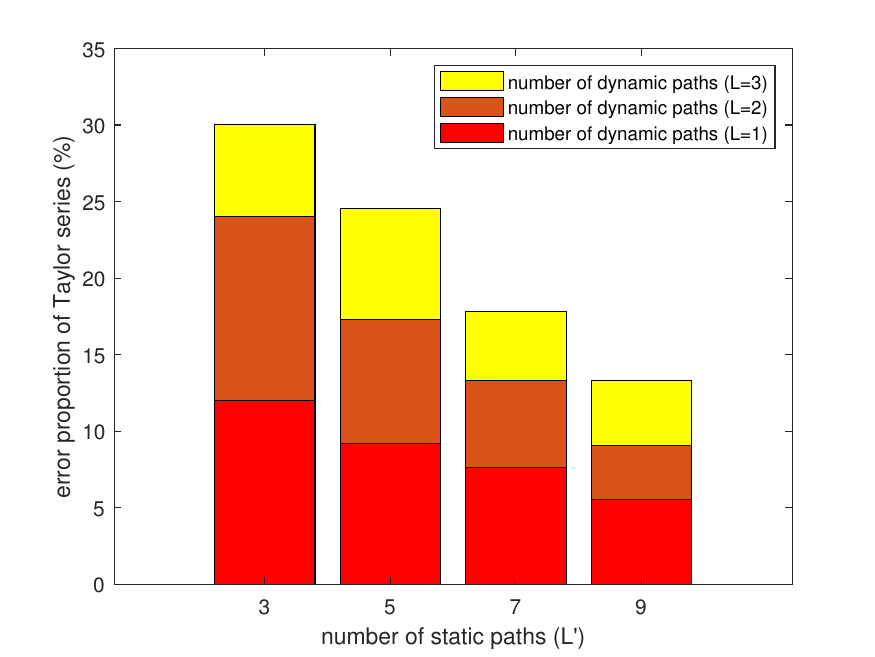} 
	}
	\subfigure[Convergence probability]{  
	\includegraphics[scale=0.5]{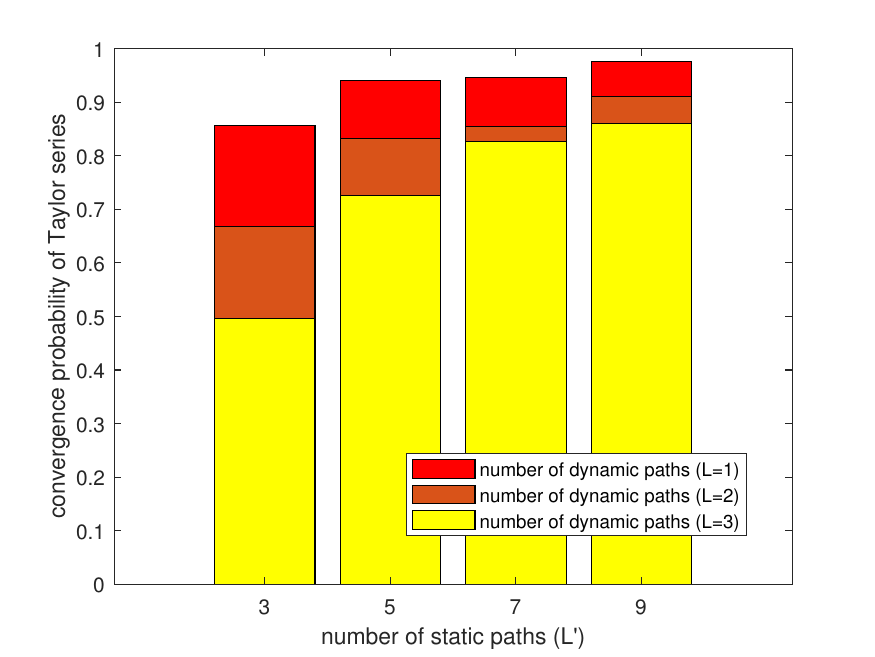}
	}
	\caption{Error proportion and convergence probability of Taylor series }
	\label{Taylor}
\end{figure}

\subsection{Numerical Results}

In numerical results, the carrier frequency is $3$ GHz. The number of subcarriers is $G=64$. The frequency bandwidth is $64$ MHz. Hence, the OFDM symbol period $T$ is 1 $\mu$s. The propagation delay is randomly distributed over $[0, 0.4]$ $\mu$s. The CP period $T_{\rm C}$ is $0.3$ $\mu$s. The approximate interval between two packets, $T_{\rm A}$, is 1 ms. We use  $M=128$ packets for the parameter estimation. The parameters remain unchanged in these packets.
 The velocity of targets ranges from -30 meter-per-second (mps) to 30 mps, and the Doppler frequency is randomly distributed over $[-0.3, 0.3]$ kHz. The AoAs of targets are random values uniformly distributed from $-\pi/2$ to $ \pi/2$. All the targets are modelled as point sources, and the radar cross-sections are assumed to be 1. The BS employs a ULA with $N=8$ antenna elements. There is one dynamic path and $L_S=5$ static paths unless mentioned individually.  The power of all paths is assumed to be equal, and hence, there would be no requirement for a dominant LOS path.

Fig. \ref{Taylor} shows error proportion and convergence probability of Taylor series versus $L$ and $L_S$. Mathmatically, the error proportation is given by
\begin{align}\label{ERR}
 e_{\rm Tay}(p)
 =& \sum\limits_{g=0}^{G-1}  \sum\limits_{m=0}^{M-1} \sum\limits_{n=1}^N   \frac 1{MGN|\xi_{n,q}[m+p,g]|^2} \left| \smallo^3(\Z) \right|^2.
\end{align}
Here, we let $p=2$.  Note that we regard the harmonic of the second order of the Taylor series as errors too. 
The error proportion can reflect the accuracy of the Taylor series. We only count in the error in which the absolute value is less than 1 w.r.t. different $m, g,$ and $n$. Otherwise,  the error is divergent and removed from the summation in \eqref{ERR}. Hence, Fig. \ref{Taylor} (b) supplements the corresponding convergence probability. From the figure, we observe that the error proportion decreases with $L_S$ increasing and with $L$ decreasing. This means that more static paths and less dynamic paths make the Taylor series convergent more quickly. Fig. \ref{Taylor} (b) also indicates the same conclusion.
In general, the number of static paths is far more than that of dynamic paths. When involving the LOS path, the power of static paths can be 10 times higher than that of the dynamic path, and hence, the expected error proportion should be around 5\% and the expected convergence probability is about 0.95.   In the worst case of $L=L_S$, which rarely happens, the convergence probability is around 0.5.
\begin{figure}[t]
    \centering
    \includegraphics[scale=0.5]{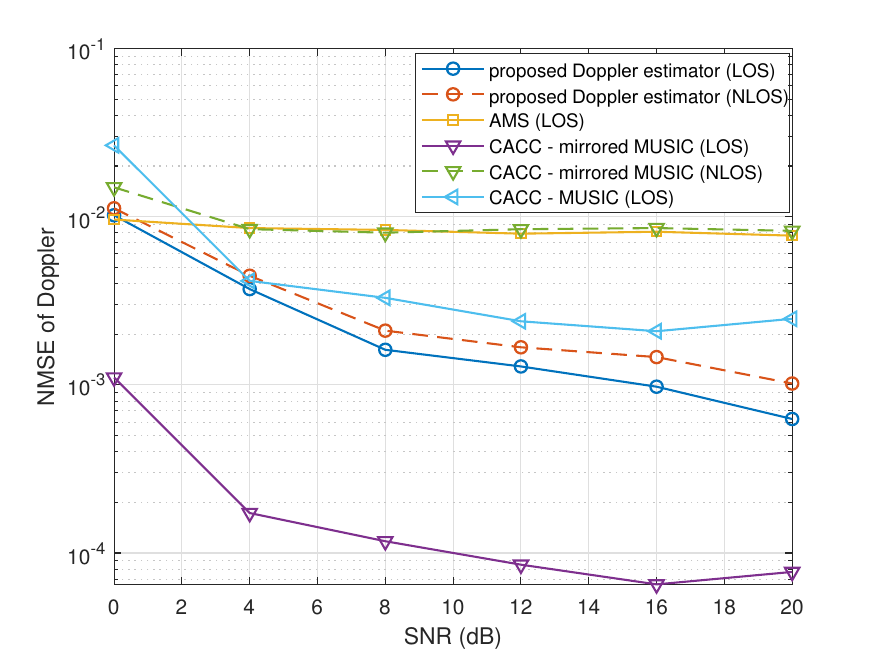}
    \caption{ NMSE versus SNR for Doppler frequency. }
    \label{NMSEDop}
\end{figure}
\begin{figure}[t]
    \centering
    \includegraphics[scale=0.5]{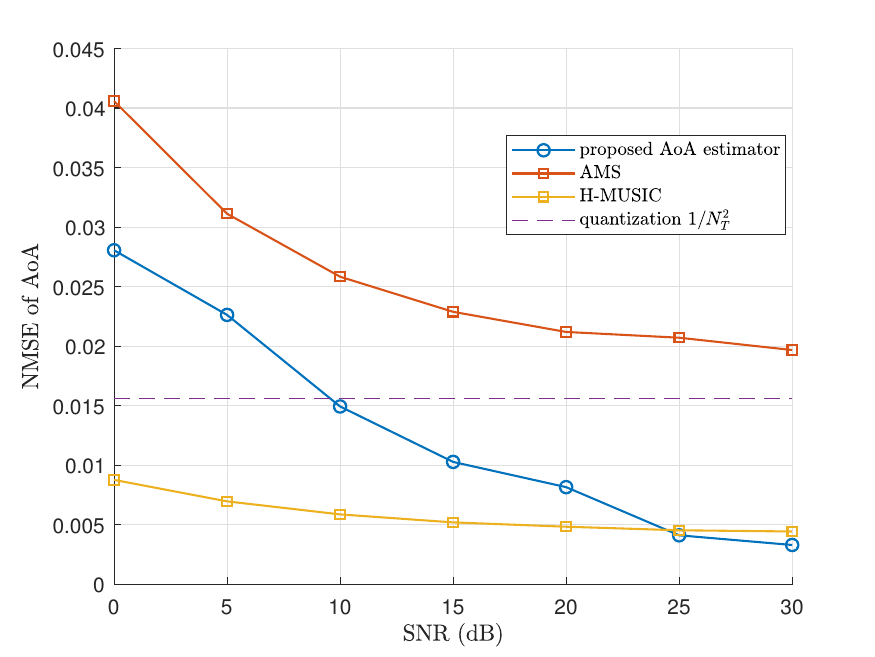}
    \caption{NMSE versus SNR for AoA.}
    \label{NMSEAoA}
\end{figure}
\begin{figure}[t]
    \centering
    \includegraphics[scale=0.5]{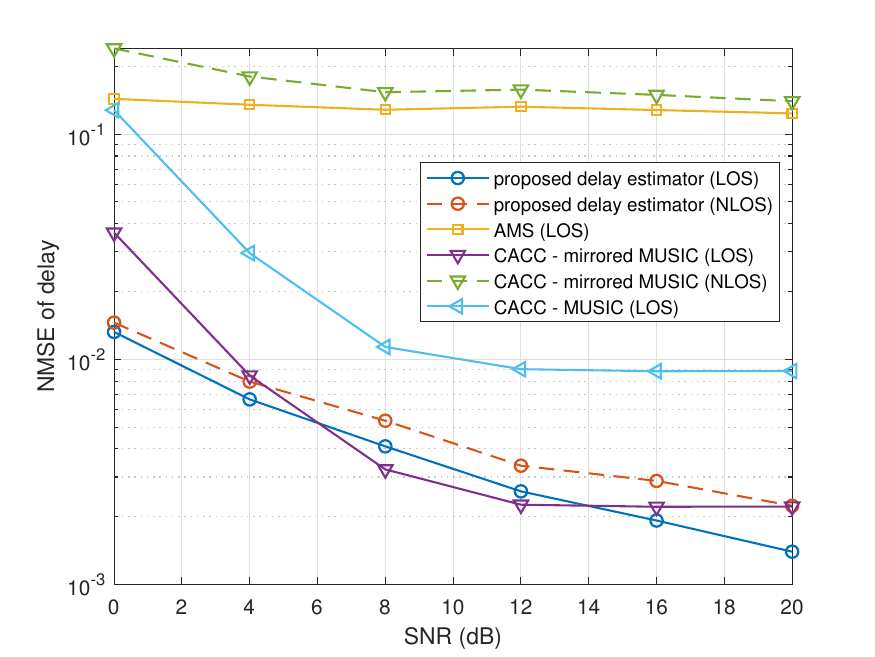}
    \caption{NMSE versus SNR for delay.}
    \label{NMSEDel}
\end{figure}
 
Fig. \ref{NMSEDop}, Fig. \ref{NMSEAoA}, and Fig. \ref{NMSEDel} plot the normalized mean-squared-error (NMSE) versus signal-to-noise ratio (SNR) of Doppler frequency, AoA, and delay, respectively.
We consider a case of $L=1$ and $L_S=5$ with and without a LOS path. The power of the LOS path is 10 dB higher than that of the NLOS path.
The parameter of our proposed estimators, $P$, equals $30$. Other system setups are given at the beginning of this subsection.
 Other uplink sensing benchmarks, which addressed the TO and CFO in asynchronous systems, are included in our simulations. 
 For Doppler and delay estimators, the benchmarks are the AMS method \cite{IndoTrack}, CACC with MUSIC, and CACC with mirrored MUSIC \cite{niTSP}. Due to the high computational complexity of the WiDFS method, it is not included in the comparisons. As for the AoA estimator, we compare our AoA estimator with the AMS, H-MUSIC \cite{high_reso}, and the average NMSE of quantized grids, i.e., $1/N_T^2 \approx 1.56\times 10^{-2}$. The figures demonstrate that our proposed estimator outperforms the AMS method for all three parameters. 

For  NMSE of Dopper frequency, our estimator is nearly the same as CACC-MUSIC in the LOS scenario.
Since the received signals involve a dominant LOS path, CACC-mirrored-MUSIC  has higher accuracy than our proposed estimator, however, 
CACC methods only work with the existence of a LOS scenario. Without the LOS path, it is seen that the NMSEs of CACC methods rise dramatically. Without the LOS path, our proposed estimator is better than all the benchmarks.

For NMSE of AoA, our method can achieve nearly the same NMSE as H-MUSIC at high SNRs. In low SNRs, the performance of our proposed AoA estimator degrades.
This is because the Taylor series is convergent when the dynamic component is little. Since the noise can be treated as the dynamic component, the Taylor series will be divergent with a high noise floor. Despite the defect, our method obtains the dynamic path separately while both H-MUSIC and CACC methods need to estimate the parameters of all NLOS paths.

For NMSE of delay, our method can achieve the best performances without the LOS path. With the existence of the LOS path,  our delay estimator still achieves the lowest NMSE at a high SNR.
Besides, CACC methods can only obtain relative delays, which means that the delay of the LOS path is necessary to be known at the BS.
Our delay estimator can obtain the absolute delays as long as $S_n[g]$ is known.
 
\begin{figure}[t]
    \centering
    \includegraphics[scale=0.6]{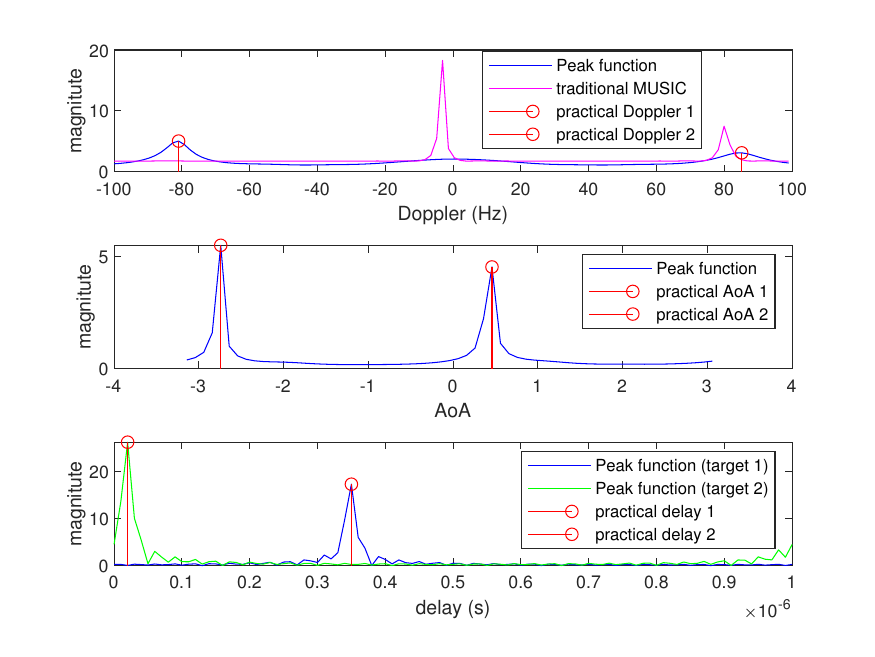}
    \caption{The shape of Peak functions of the multi-path estimator.}
    \label{peak}
\end{figure}

Fig. \ref{peak} illustrates the shapes of the `Peak' functions with considering multiple dynamic paths in a noiseless environment. We let $L=2$ and $L_S=5$. Other system parameters are the same as Fig. \ref{NMSEDop}. For the Doppler frequency,  we can observe two peaks tightly match with the practical Doppler frequencies. We compare its shape with the traditional MUSIC which has a large peak around 0 Hz due to the existence of static components. In the traditional MUSIC, the estimated Doppler frequencies have unknown CFOs mixed with practical values. Even worse, it is noted that some dynamic paths are missing in the peaks of traditional MUSIC.  
For the AoA estimator,  the peaks have a sharp shape and match with the practical AoAs tightly. Traditional MUSIC cannot separate dynamic AoAs from static ones and hence would generate $7$ peaks, which is not shown in this figure.
For the delay estimator, we see that each peak occupies the entire range of $[0, T]$ individually, which is because the delays are obtained from different rows of the LS outputs in \eqref{LSout}. This means that the delays do not have ambiguity problems when multiple delays are close to each other.

\section{Conclusion}
This paper has proposed a Doppler-AoA-delay estimation scheme under the uplink ISAC systems, where only the parameters of moving targets are estimated. The system does not require synchronization between transceivers thanks to the CSI ratio. Our novel estimation scheme is mainly based on the Taylor series of the CSI ratio, which shows a good convergence property and makes it possible to transform the non-linear CSI ratio into linear forms. The simulation results show that the performance of our scheme becomes better with a larger power of static components. The proposed AoA and delay estimators outperform the benchmark at high SNR. The Doppler estimator outperforms the benchmarks in the NLOS scenario.  Our work can be effectively applied in the PMN and WiFi sensing networks.

\begin{appendices}
\section{Convergence of Taylor Series of CSI Ratio}\label{App1}
From \eqref{1thder} and \eqref{2thder}, we can generalize the $k$th derivative of $f(\Z_m)$ as \eqref{guinak}, where $y_n$ is short for $y_n[m,g]$, $y_{n-q}$ is short for $ y_{n-q}[m,g]$, $ e^{j k \bar f_{\rm O}[m] } =e^{j k mT_{\rm A}  f_{\rm O}[m]}$, and $ e^{-j k  \bar\tau_{\rm O}[m] } =e^{-j k \frac g T  \tau_{\rm O}[m]}$. 
Both derivatives w.r.t. $k=1$ and $k=2$ satisfy  \eqref{guinak}. The higher-order Taylor series can be proved by using the induction method, while omitted due to the page limits.
 
\begin{table*}  
\begin{align}\label{guinak}
f^{(k)}(z_{l_1},\cdots,z_{l_k})= (-1)^k \frac{ k! e^{-j q  \sum\limits_{i=1}^k \p_{l_i} } y_n -  (k-1)! \sum\limits_{i=1}^k e^{-jq \sum\limits_{\mathop {h=1}\limits_{h\neq i} }^{k}\p_{l_h} } y_{n-q} }{y_{n-q} ^{ k+1 }}e^{j k \bar f_{\rm O}[m] }e^{-j k \bar \tau_{\rm O}[m] }, 
l_k\in\{ 1, \cdots, L\} 
\end{align} 
\end{table*}

Then, the $k$th individual coefficient of the Taylor series is given by
\begin{align}\label{termk}
  \frac{1}{k!}f^{(k)}(z_{l_1},\cdots,z_{l_k}) \prod\limits_{i=1}^k(z -z_{l_i}) .
\end{align}
We need to prove that the absolute value of \eqref{termk} decreases with $k$ increasing. Note that
\begin{align} \label{termup}
& \frac{1}{k!}f^{(k)}(z_{l_1},\cdots,z_{l_k}) \prod\limits_{i=1}^k(z -z_{l_i})\notag\\
\leq&    \frac{    |y_n | + |y_{n-q}|  }{|y_{n-q}|^{ k+1 }} \prod\limits_{i=1}^k|z -z_{l_i}| 
= \frac{    |y_n | + |y_{n-q}|  }{\left| y_{n-q} \right|^{ k+1 }} \prod\limits_{i=1}^k|z -z_{l_i}|\notag\\
 \leq &\frac{    |y_n | + |y_{n-q}|  } {|y_{n-q}|}  \frac  { \prod\limits_{i=1}^k|z-z_{l_i}| }{\left|y_{n-q} \right|^{ k }} 
  \leq  \frac{    |y_n | + |y_{n-q}|  } {|y_{n-q}|}  \frac{ \mathop{\max}\limits_l |2\alpha_l|^{k} } {\left| y_{n-q} \right|^{ k }}.
\end{align}
The last term in \eqref{termup} gives the upper bound of the  $k$th coefficient of the Taylor series w.r.t. different $l_{i}$. Acorrding to the definition of Taylor series, the $k$th coefficient of the Taylor  series sums $L^k$ individual coefficients  of $l_{i}, 1\leq i\leq k$. Hence, the overall upper bound equals $\frac{    |y_n | + |y_{n-q}|  } {|y_{n-q}|}  \frac{ \mathop{\max}\limits_l |2L\alpha_l|^{k} } {\left| y_{n-q} \right|^{ k }}$.  Note that the upper bound approaches to zero with $k \rightarrow \infty$ when $|y_{n-q}|\geq \mathop{\max}\limits_l 2L|\alpha_l| $.    Therefore,    the Taylor series of CSI ratio are convergent when $|y_{n-q}|\geq \mathop{\max}\limits_l 2L|\alpha_l| $.

\section{A Trival Solution in AoA Estimator}\label{AppB} 
When $\p=0$ and $N$ is small, it can be assumed that $S_{n_1}[g]=S_{n_2}[g]\triangleq S[g]$, $\forall n_1, \forall n_2$. 
Referring to \eqref{1thder}, and abbreviating  $D_n[m,g]$,  $S_n[g]$, and  $S_n[g]+D_n[m,g]$, as  $D_n$,  $S$, and ${Y_n}$ respectively, ${\bf d}_1(0,n')$ in \eqref{D1AoA} becomes 
\begin{align}\label{D1org}
{\bf d}_1(0,n') & = [h_{0,-n'}^{m,g} (0),  \cdots, h_{0+N-1,-n'+N-1}^{m,g} (0) ]^T \notag\\
               & =\left[\frac{Y_{n'}-Y_0}{{Y}_{n'}^2},  \cdots,\frac{Y_{n'}-Y_{ N-1 }}{{Y}_{n'}^2}\right]^T  \notag\\
               & = \left[\frac{D_{n'}-D_0}{{Y}_{n'}^2},  \cdots, \frac{D_{n'}-D_{N-1}}{{Y}_{n'}^2} \right]^T,
\end{align} 
and ${\bf d}_2(0,0,n')$ becomes 
\begin{align}\label{D2org}
&{\bf d}_2(0,0,n')\notag\\ 
& =[H_{0,-n'}^{m,g} (0,0),   \cdots, H_{0+N-1,-n'+N-1}^{m,g} (0,0) ]^T  \notag\\
                 & = 2\left[\frac{D_0-D_{n'}}{{Y}_{n'}^3}, \frac{D_{ 1}-D_{n'}}{{Y}_{n'}^3}, \cdots,\frac{ D_{ N-1 }-D_{n'}}{{Y}_{n'}^3}\right]^T,
\end{align}

Given that $L=1$ and abbreviating  $S_n[g]+D_n[m+p,g]$ as $Y_n (p)$,  the $(n'+1)$th column of  ${\bf A}[m,p,g]$ in \eqref{AAA} is 
\begin{align}\label{vecAA}
&{\rm v}_{n'}({\bf A}[m,p,g]) \notag\\
= &  [\psi_{0,-n'},\cdots, \psi_{N-1,-n'+N-1}]^T  \notag\\ 
= & \left[\frac{Y_0(p)}{Y_{n'}(p)}-\frac{Y_0 }{Y_{n'} },\cdots,\frac{Y_{N-1}(p)}{Y_{n' }(p)}-\frac{Y_{N-1} }{Y_{n'} }\right]^T\notag \\
 =& \left[\frac{S +e^{j2\pi pT_{\rm A}f\dop} D_0} {S +e^{j2\pi pT_{\rm A}f\dop}D_{ n'} }-\frac{S + D_{0}} {S + D_{ n'} },\cdots,\right.\notag\\
 &\left.\frac{S +e^{j2\pi pT_{\rm A}f\dop}D_{N-1}} {S +e^{j2\pi pT_{\rm A}f\dop}D_{n'} }-\frac{S + D_{N- 1}} {S + D_{n'} }\right]^T   \notag\\ 
 = &  \left[\frac{S(D_{n'}-D_{0})(1-e^{j2\pi pT_{\rm A}f\dop} ) } {Y_{n'}(S +e^{j2\pi pT_{\rm A}f\dop}D_{n'}) } ,\cdots,\right.\notag\\
 &\left.\frac{S(D_{n'}-D_{N-1})(1-e^{j2\pi pT_{\rm A}f\dop}) } {Y_{n'}(S +e^{j2\pi pT_{\rm A}f\dop}D_{n'}) } \right]^T,
\end{align}
 where ${\rm v}_{n'}(\cdot)$ is the $n'$th column of a matrix.
 
To prove that $\p=0$ is a trivial solution to \eqref{AoAP}, it is equivalent to proving that ${\bf d}_1(0,n')$ and ${\bf d}_2(0,0,n')$ are the basis vectors of $ {\bf A}[m,p,g]$.
It is clear that \eqref{vecAA} is linearly dependent with \eqref{D1org} and \eqref{D2org}. Hence, the rank of $[{\bf d}_1(0,n'),d_2(0,0,n'), {\bf A}[m,p,g]]$ is 1. Then, noting that each column of $ \bar{\bf A}$ is
\begin{align} 
 \left[ 
 {\rm v}_0({\bf A}[m,p,g]),  \cdots  ,  {\rm v}_{N-1}({\bf A}[m,p,g])    \right],
\end{align}
  ${\bf d}_1(0) = [{\bf d}_1^T(0 , 0), \cdots, {\bf d}_1^T(0,N-1)]^T$, and  ${\bf d}_2(0,0) = [{\bf d}_2^T(0 ,0, 0), \cdots, {\bf d}_2^T(0,0,N-1)]^T$, we have 
\begin{align}
{\rm rank}(\bar{\bf A}) ={\rm rank}([\bar{\bf A},{\bf d}_1(0),{\bf d}_2(0,0)])=N.
\end{align}  
Note that ${\bf d}_1(0)$ and ${\bf d}_2(0,0)$ are the basis vectors for any given $ \bar{\bf A}$. Likewise, the higher-order vectors of ${\bf d}_i, i\geq3,$ are also the basis vectors for any given $ \bar{\bf A}$.
Therefore, $\p=0$ is a trivial solution to \eqref{AoAP} when $S_{n_1}[g]=S_{n_2}[g]$ and $L=1$, which ends the proofs.
 \end{appendices}
 


\begin{thebibliography}{10}
\providecommand{\url}[1]{#1}
\csname url@samestyle\endcsname
\providecommand{\newblock}{\relax}
\providecommand{\bibinfo}[2]{#2}
\providecommand{\BIBentrySTDinterwordspacing}{\spaceskip=0pt\relax}
\providecommand{\BIBentryALTinterwordstretchfactor}{4}
\providecommand{\BIBentryALTinterwordspacing}{\spaceskip=\fontdimen2\font plus
\BIBentryALTinterwordstretchfactor\fontdimen3\font minus
  \fontdimen4\font\relax}
\providecommand{\BIBforeignlanguage}[2]{{%
\expandafter\ifx\csname l@#1\endcsname\relax
\typeout{** WARNING: IEEEtran.bst: No hyphenation pattern has been}%
\typeout{** loaded for the language `#1'. Using the pattern for}%
\typeout{** the default language instead.}%
\else
\language=\csname l@#1\endcsname
\fi
#2}}
\providecommand{\BIBdecl}{\relax}
\BIBdecl

\bibitem{lushan}
M.~L. {Rahman}, J.~A. {Zhang}, X.~{Huang}, Y.~J. {Guo}, and R.~W. {Heath},
  ``Framework for a perceptive mobile network using joint communication and
  radar sensing,'' \emph{IEEE Trans. Aerosp. Electron. Syst.}, vol.~56, no.~3,
  pp. 1926--1941, Jun. 2020.

\bibitem{lushanSvy}
J.~A. Zhang, M.~L. Rahman, K.~Wu, X.~Huang, Y.~J. Guo, S.~Chen, and J.~Yuan,
  ``Enabling joint communication and radar sensing in mobile networks—a
  survey,'' \emph{IEEE Communications Surveys \& Tutorials}, vol.~24, no.~1,
  pp. 306--345, 2022.

\bibitem{kumari}
P.~{Kumari}, J.~{Choi}, N.~{Gonz\'{a}lez-Prelcic}, and R.~W. {Heath}, ``{IEEE}
  802.11ad-based radar: {An} approach to joint vehicular communication-radar
  system,'' \emph{IEEE Trans. Veh. Technol.}, vol.~67, no.~4, pp. 3012--3027,
  Apr. 2018.

\bibitem{strum1}
C.~{Sturm} and W.~{Wiesbeck}, ``Waveform design and signal processing aspects
  for fusion of wireless communications and radar sensing,'' \emph{Proc. IEEE},
  vol.~99, no.~7, pp. 1236--1259, Jul. 2011.

\bibitem{abu2018performance}
Z.~Abu-Shaban, X.~Zhou, T.~Abhayapala, G.~Seco-Granados, and H.~Wymeersch,
  ``Performance of location and orientation estimation in {5G} mmwave systems:
  Uplink vs downlink,'' in \emph{2018 IEEE Wireless Communications and
  Networking Conference (WCNC)}.\hskip 1em plus 0.5em minus 0.4em\relax IEEE,
  2018, pp. 1--6.

\bibitem{huang2021mimo}
S.~Huang, M.~Zhang, Y.~Gao, and Z.~Feng, ``Mimo radar aided mmwave time-varying
  channel estimation in {MU-MIMO V2X} communications,'' \emph{IEEE Trans.
  Wireless Commun.}, vol.~20, no.~11, pp. 7581--7594, 2021.

\bibitem{li2021outer}
C.~Li, N.~Raymondi, B.~Xia, and A.~Sabharwal, ``Outer bounds for a joint
  communicating radar (comm-radar): {The} uplink case,'' \emph{IEEE Trans.
  Commun.}, vol.~70, no.~2, pp. 1197--1213, 2021.

\bibitem{zheng2017super}
L.~Zheng and X.~Wang, ``Super-resolution delay-doppler estimation for {OFDM}
  passive radar,'' \emph{IEEE Trans. Signal Process.}, vol.~65, no.~9, pp.
  2197--2210, 2017.

\bibitem{passive10}
C.~R. {Berger}, B.~{Demissie}, J.~{Heckenbach}, P.~{Willett}, and S.~{Zhou},
  ``Signal processing for passive radar using {OFDM} waveforms,'' \emph{IEEE J.
  Sel. Topics Signal Process.}, vol.~4, no.~1, pp. 226--238, Feb. 2010.

\bibitem{ali2020leveraging}
A.~Ali, N.~Gonzalez-Prelcic, R.~W. Heath, and A.~Ghosh, ``Leveraging sensing at
  the infrastructure for {mmWave} communication,'' \emph{IEEE Commun. Mag.},
  vol.~58, no.~7, pp. 84--89, 2020.

\bibitem{garcia2016location}
N.~Garcia, H.~Wymeersch, E.~G. Str{\"o}m, and D.~Slock, ``Location-aided
  mm-wave channel estimation for vehicular communication,'' in \emph{2016 IEEE
  17th International Workshop on Signal Processing Advances in Wireless
  Communications (SPAWC)}.\hskip 1em plus 0.5em minus 0.4em\relax IEEE, 2016,
  pp. 1--5.

\bibitem{friedlander2007waveform}
B.~Friedlander, ``Waveform design for {MIMO} radars,'' \emph{IEEE Trans.
  Aerosp. Electron. Syst.}, vol.~43, no.~3, pp. 1227--1238, 2007.

\bibitem{Cataldo20}
D.~Cataldo, L.~Gentile, S.~Ghio, E.~Giusti, S.~Tomei, and M.~Martorella,
  ``Multibistatic radar for space surveillance and tracking,'' \emph{IEEE
  Aerospace and Electronic Systems Magazine}, vol.~35, no.~8, pp. 14--30, 2020.

\bibitem{caoning20}
N.~Cao, Y.~Chen, X.~Gu, and W.~Feng, ``Joint bi-static radar and communications
  designs for intelligent transportation,'' \emph{IEEE Trans. Veh. Technol.},
  vol.~69, no.~11, pp. 13\,060--13\,071, 2020.

\bibitem{zhang2010direction}
X.~Zhang, L.~Xu, L.~Xu, and D.~Xu, ``Direction of departure {(DOD)} and
  direction of arrival {(DOA)} estimation in {MIMO} radar with
  reduced-dimension {MUSIC},'' \emph{IEEE Commun. Lett.}, vol.~14, no.~12, pp.
  1161--1163, 2010.

\bibitem{gudelay}
J.~{Gu}, J.~{Moghaddasi}, and K.~{Wu}, ``Delay and {Doppler} shift estimation
  for {OFDM}-based radar-radio {(RadCom)} system,'' in \emph{2015 IEEE
  International Wireless Symposium (IWS 2015)}, Mar. 2015, pp. 1--4.

\bibitem{mehanna2015maximum}
O.~Mehanna and N.~D. Sidiropoulos, ``Maximum likelihood passive and active
  sensing of wideband power spectra from few bits,'' \emph{IEEE Trans. Signal
  Process.}, vol.~63, no.~6, pp. 1391--1403, 2015.

\bibitem{kong2018joint}
B.~Kong, Y.~Wang, X.~Deng, and D.~Qin, ``Joint range-doppler-angle estimation
  for {OFDM}-based {RadCom} system via tensor decomposition,'' \emph{Wireless
  Communications and Mobile Computing}, 2018.

\bibitem{sanson2020high}
J.~B. Sanson, P.~M. Tom{\'e}, D.~Castanheira, A.~Gameiro, and P.~P. Monteiro,
  ``High-resolution delay-doppler estimation using received communication
  signals for {OFDM} radar-communication system,'' \emph{IEEE Trans. Veh.
  Technol.}, vol.~69, no.~11, pp. 13\,112--13\,123, 2020.

\bibitem{hyder2016zadoff}
M.~Hyder and K.~Mahata, ``Zadoff--chu sequence design for random access initial
  uplink synchronization in {LTE}-like systems,'' \emph{IEEE Trans. Wireless
  Commun.}, vol.~16, no.~1, pp. 503--511, 2016.

\bibitem{IndoTrack}
\BIBentryALTinterwordspacing
X.~Li, D.~Zhang, Q.~Lv, J.~Xiong, S.~Li, Y.~Zhang, and H.~Mei, ``Indotrack:
  {Device}-free indoor human tracking with commodity {Wi-Fi},'' \emph{Proc. ACM
  Interact. Mob. Wearable Ubiquitous Technol.}, vol.~1, no.~3, Sep. 2017.
  [Online]. Available: \url{https://doi.org/10.1145/3130940}
\BIBentrySTDinterwordspacing

\bibitem{widar2.0}
\BIBentryALTinterwordspacing
K.~Qian, C.~Wu, Y.~Zhang, G.~Zhang, Z.~Yang, and Y.~Liu, ``Widar2.0: Passive
  human tracking with a single {Wi-Fi} link,'' in \emph{Proceedings of the 16th
  Annual International Conference on Mobile Systems, Applications, and
  Services}.\hskip 1em plus 0.5em minus 0.4em\relax New York, NY, USA:
  Association for Computing Machinery, 2018, pp. 350--361. [Online]. Available:
  \url{https://doi.org/10.1145/3210240.3210314}
\BIBentrySTDinterwordspacing

\bibitem{wang2022single}
Z.~Wang, J.~A. Zhang, M.~Xu, and J.~Guo, ``Single-target real-time passive
  {WiFi} tracking,'' \emph{IEEE Trans. Mobile Comput.}, 2022.

\bibitem{niTSP}
Z.~Ni, J.~A. Zhang, X.~Huang, K.~Yang, and J.~Yuan, ``Uplink sensing in
  perceptive mobile networks with asynchronous transceivers,'' \emph{IEEE
  Trans. Signal Process.}, vol.~69, pp. 1287--1300, 2021.

\bibitem{zhang2022integration}
J.~A. Zhang, K.~Wu, X.~Huang, Y.~J. Guo, D.~Zhang, and R.~W. Heath,
  ``Integration of radar sensing into communications with asynchronous
  transceivers,'' \emph{IEEE Commun. Mag.}, 2022.

\bibitem{zeng2019farsense}
Y.~Zeng, D.~Wu, J.~Xiong, E.~Yi, R.~Gao, and D.~Zhang, ``Farsense: Pushing the
  range limit of {WiFi}-based respiration sensing with {CSI} ratio of two
  antennas,'' \emph{Proceedings of the ACM on Interactive, Mobile, Wearable and
  Ubiquitous Technologies}, vol.~3, no.~3, pp. 1--26, 2019.

\bibitem{zeng2020multisense}
Y.~Zeng, D.~Wu, J.~Xiong, J.~Liu, Z.~Liu, and D.~Zhang, ``Multisense: Enabling
  multi-person respiration sensing with commodity wifi,'' \emph{Proceedings of
  the ACM on Interactive, Mobile, Wearable and Ubiquitous Technologies},
  vol.~4, no.~3, pp. 1--29, 2020.

\bibitem{Laoudias2018}
C.~Laoudias, A.~Moreira, S.~Kim, S.~Lee, L.~Wirola, and C.~Fischione, ``A
  survey of enabling technologies for network localization, tracking, and
  navigation,'' \emph{IEEE Communications Surveys \& Tutorials}, vol.~20,
  no.~4, pp. 3607--3644, 2018.

\bibitem{strum13}
C.~{Sturm}, Y.~L. {Sit}, M.~{Braun}, and T.~{Zwick}, ``Spectrally interleaved
  multi-carrier signals for radar network applications and multi-input
  multi-output radar,'' \emph{IET Radar, Sonar Navigation}, vol.~7, no.~3, pp.
  261--269, Mar. 2013.

\bibitem{sit}
Y.~L. {Sit}, B.~{Nuss}, and T.~{Zwick}, ``On mutual interference cancellation
  in a {MIMO} {OFDM} multiuser radar-communication network,'' \emph{IEEE Trans.
  Veh. Technol.}, vol.~67, no.~4, pp. 3339--3348, Apr. 2018.

\bibitem{high_reso}
S.~{Chuang}, W.~{Wu}, and Y.~{Liu}, ``High-resolution {AoA} estimation for
  hybrid antenna arrays,'' \emph{IEEE Trans. Antennas Propag.}, vol.~63, no.~7,
  pp. 2955--2968, Jul. 2015.

\end{thebibliography}
\end{document}